\documentclass[submission,copyright,creativecommons]{eptcs}
\providecommand{\event}{SYNT 2014}

\usepackage{mathtools}
\usepackage{amsthm}
\usepackage{amssymb}
\usepackage{bbm}
\usepackage{cite}
\usepackage{tikz}
\usepackage{paralist}
\usepackage{fixme,manfnt,mparhack}
\usepackage[basic]{complexity}
\usepackage[T1]{fontenc}
\usepackage{enumerate}
\usepackage[ruled,linesnumbered]{algorithm2e}
\usepackage{xspace}
\usepackage{macros}
\usepackage{longtable}
\usepackage{wrapfig}

\title{AbsSynthe: abstract synthesis from succinct safety
specifications\thanks{This work was supported by the ERC inVEST (279499)
project.}}

\author{Romain Brenguier, Guillermo A. P\'{e}rez,
Jean-Fran\c{c}ois Raskin \and Ocan Sankur
\institute{Universit\'{e} Libre de Bruxelles -- Brussels, Belgium}
\email{\{rbrengui,gperezme,jraskin,osankur\}@ulb.ac.be}
}

\def\titlerunning{AbsSynthe -- an abstract synthesis tool.}
\def\authorrunning{R. Brenguier et al.}

\begin{document}

\maketitle

\begin{abstract}
In this paper, we describe a synthesis algorithm for safety specifications
described as circuits. Our algorithm is based on fixpoint computations,
abstraction and refinement, it uses binary decision diagrams as symbolic data
structure. We evaluate our tool on the benchmarks provided by the organizers of
the synthesis competition organized within the SYNT'14 workshop.  
\end{abstract}

\section{Introduction}

The {\em model-checking} approach to verification of reactive systems is as
follows. Given a model of the system together with a description of the
environment in which it is embedded, and a specification that formalizes a
property of interest of the system, an algorithm verifies that all the possible
behaviors of the system within its environment comply with the specification.
Model-checking has been proposed in the eighties and is now a standard technique
to improve the reliability of reactive systems. 

{\em Synthesis} goes a step further: synthesis only requires a model of the environment together with a specification of a property that the system must enforce within the environment, it does not require a model of the system. From the description of the environment and the property, an algorithm tries to build automatically a system that is correct by construction, i.e. a system which enforces the specification. If such a system does not exist then the synthesis algorithm can also provide feedback in the form of a strategy for the environment that enforces the negation of the specification and so shows why the specification cannot be realized.  

The synthesis problem can be formalized as a two-player game on a graph with an
omega-regular objective.  While the theory that underlines those games is now
well understood, see e.g.~\cite{thomas95}, there are only a few implementations
available~\cite{ehlers10,fjr09,ss09,bks14} and the sizes of systems
on which synthesis has been applied are usually much smaller than the sizes of systems for which model-checking has been successfully applied.

This paper describes our experiences with building\GP{Changed ``attempts'' to
experiences}
a prototype of tool to participate to
the first synthesis competition organized within the SYNT'14 workshop. The set
up for this competition is as follows. Given a sequential circuit description
with exactly one output signal and a partition of its input signals into
\emph{controllable} inputs that belongs to the system (to synthesize) and
\emph{uncontrollable} inputs that belong to the environment,
decide if there is a strategy to choose the controllable input signals such that no
matter how the uncontrollable input signals are updated along the execution of
the circuit, the output of the circuit is always low. In terms of two player
games, the winning condition (the specification) for the system is thus a safety
objective. If a winning strategy exists for the system, then build a circuit
which implements this strategy. 

The realizability problem for safety specifications is known to be solvable in
linear time with respect to the size of the underlying game graph (see,
e.g.~\cite{gradel04,thomas95}). However, here the underlying graph is given
implicitly and succinctly by the circuit description and in this case the
problem is known to be complete for \EXP~(see, e.g. \cite{py86}). To combat the
state explosion problem, we adapt two classical techniques that have proven
useful in the context of model-checking: we use binary decision
diagrams~\cite{bryant86} as a data structure to represent and manipulate
symbolically sets of configurations of the circuit specification, and we use
abstraction and refinement to simplify the underlying game and lower its
dimension (the number of Boolean variables that are necessary for its
description). The abstraction refinement algorithm that we have defined can be
seen as combining the ideas of~\cite{ap07} and~\cite{hjm03} together with binary
decision diagrams (BDDs) and some additional heuristics, all this formalized
with abstract interpretation as in~\cite{hmmr00}.

We have implemented a fixpoint computation together with several optimizations
that lead to a synthesis algorithm that is able to handle circuits with several
tens of latches and a few hundreds of gates. We report on the experiments that
we have conducted on all the benchmarks provided by the organizers of the
synthesis competition that were available at the time of submission.
In a vast majority of the benchmarks the best performing\GP{Removed
``mitigated'' sentence} version of our algorithm is the {\em plain fixpoint} algorithm that does not use
abstraction at all. However, to be efficient, the explicit construction of the
BDDs for the transition relation needs to be avoided~\cite{bcl91}, our solution
is to use substitution of variables with BDDs as
in~\cite{cm90} to directly compute the effect of the
transition relation backwards. Nevertheless, for some examples, abstraction and
refinement are necessary: our algorithm based on abstraction and refinement
terminates while the basic concrete fixpoint computation does not. 
We think that
the lack of good performance of abstraction in our experiments is partly due to
the fact that there is no explicit structure in
the circuit description on which we apply our analysis.
In fact, we consider circuits given in a low-level description, which is neither
hierarchical nor compositional, so
the usual techniques used in program verification~\cite{cgjlv00}
does not seem to be applicable here.
Another reason could also be that
the benchmarks considered here are control intensive and not data oriented where
abstraction seems to perform better~\cite{cgt04}.  We strongly believe that more
research is necessary for understanding how to recover interesting aspects of
the structure present in the circuits from their low level descriptions and use
this structure in the abstraction procedure.  Finally, we also report on how to
synthesize a circuit from the winning region computed by our algorithm and how
to exploit reachability information to decrease the size of the synthesized
circuit with minimization of BDDs using don't care sets~\cite{hbbm97}.

\paragraph*{Related works} 
Recent efforts to find efficient algorithms for synthesis have been reported
in~\cite{bks14} where solutions based on solvers for QBF and SAT are investigated.
In that paper, the authors compare their solutions with a BDD implementation
that constructs explicitly the transition relation. The conclusions of their
paper resemble our conclusions: the BDD implementation usually outperforms the
QBF-SAT algorithms with the exception of a few examples. Our BDD implementation
that does not construct the transition relation is usually largely more
efficient than the one that constructs the BDD for the transitions relation.

In~\cite{nlbrw14} the authors present an algorithm for synthesis which searches
for a small set of plays that witness a winning strategy for one of the players.
They report their tool works well for games in which winning strategies admit
compact representations.

The problem of minimization of Boolean functions used in circuit constructions
has been widely studied in the logic synthesis community~(e.g.
\cite{mccluskey56,hbbm97}).

\section{Preliminaries}
We will present our algorithms in set-theoretic notation.  However, in order to
provide symbolic representations of sets and the implementation of set operators
we will also represent sets by Boolean functions, and use both notations
interchangeably.  Formally, we let $\mathbb{B} = \{0,1\}$, and if~$L$ denotes a
finite set of variables, a function $L \to \mathbb{B}$ is called a
\emph{valuation of $L$}. Note that a valuation~$v$ defines a subset $v^{-1}(1)$
of~$L$.  We will also consider Boolean functions $\mathbb{B}^L \rightarrow
\mathbb{B}$ to denote sets of valuations.

We will describe Boolean functions by \emph{first-order logic} formulas on a
given set of variables~$V$, which are made of propositional logic and
first-order quantification on~$V$.
A~\emph{formula}~$f$ whose free variables are~$X$ will be written~$f(X)$.
If the free variables are~$X\cup Y$ for two sets~$X,Y$, we may also write
$f(X,Y)$.
When we quantify over a set of variables~$L$, we will write
$\exists L$ instead of $\exists l_1\exists l_2 \ldots \exists l_n$ if
$L=\{l_1,\ldots,l_n\}$, and similarly for universal quantification.

%

Let $X,Y,Z$ be three sets of variables such that $Y\subseteq X$ and $X \cap Z = \varnothing$.  
Consider a formula $f(X)$ and a set of formulas $(g_y(Z))_{y \in Y}$ (one for each element in $Y$).  
We denote by $f[y \leftarrow g_y]_{y \in Y}$ the formula $f$ in which every $y \in Y$ has been substituted by the corresponding $g_y$.
Formally, $f[y \leftarrow g_y]_{y \in Y} (X \setminus Y,Z) = \exists Y.\ f(X) \land \left(\bigwedge_{y\in Y} y \Leftrightarrow g_y(Z)\right)$.
This work has been implemented using BDDs~\cite{bryant86} to perform
all
symbolic operations on Boolean functions.  BDD packages provide optimized
procedures to do substitution (see e.g. function \emph{compose}
in~\cite{somenzi99}).

\paragraph*{Circuit specifications}
We are interested in synthesizing controllers for synchronous sequential
circuits enforcing a given safety specification, where some inputs are
\emph{controllable}, and others are \emph{uncontrollable}. Intuitively,
controllable inputs are to be determined by the controller to be synthesized,
while uncontrollable inputs cannot be restricted, and are determined by the
environment. 
A distinguished latch indicates if an error has
occurred. Formally, a \emph{synchronous sequential circuit} is a tuple $\langle
X_u, X_c, L, (f_l)_{l\in L}, \fbad \rangle$, where:
\begin{itemize}
	\item $X_u, X_c, L$ are finite sets of boolean variables representing
		\emph{uncontrollable inputs}, \emph{controllable inputs}, and
		\emph{latches} respectively;
	\item for each latch $l\in L$, $f_l \colon \mathbb{B}^{X_u} \times
		\mathbb{B}^{X_c} \times \mathbb{B}^L \to \mathbb{B}$ is the
		\emph{transition function} that gives the valuation of $l$ in
		the next step; 
	\item $\fbad$ is the \emph{error function} $\fbad \colon
		\mathbb{B}^{X_u} \times \mathbb{B}^{X_c} \times \mathbb{B}^L \to
		\mathbb{B}$, which evaluates to true in error states. 
\end{itemize}

Given a circuit, our goal is to synthesize a \emph{controller} which, given any
valuation of the latches and uncontrollable inputs, sets the controllable
inputs, in order to ensure that the overall system never enters an error state.

We assume that
\begin{inparaenum}[(i)]
\item there is a latch $\out \in L$ which, once it becomes true, stays true, and
\item that the latches are initialized to $0$, i.e.
the initial valuation is $v(l) = 0$ for all~$l$.
\end{inparaenum}

\paragraph*{Reachability and Safety Games}
The problem of controller synthesis can be formalized as a game between two
players, namely, \adam and \eve, played on a graph~(see, e.g.~\cite{thomas95}).
Formally, an \emph{arena} is a tuple $G = \langle Q, q_I, \Sigma_u, \Sigma_c,
\Delta \rangle$ where:
\begin{inparaenum}[(i)]
	\item $Q$ is a finite set of states;
	\item $q_I \in Q$ is the initial state;
	\item $\Sigma_u$ is a finite set of uncontrollable actions;
	\item $\Sigma_c$ a finite set of controllable actions;
	\item $\Delta \subseteq Q \times \Sigma_u \times \Sigma_c \times Q$ is a
		transition relation.
\end{inparaenum}

The game is initially in state~$q_I$ and is played in rounds. At every round,
from state $q$, \adam chooses an action~$a_u$ from $\Sigma_u$ and \eve responds
by choosing an action~$a_c$ from $\Sigma_c$ and a successor state $s \in Q$ such
that $(q,a_u,a_c,s) \in \Delta$. We write $\delta$ instead of $\Delta$ if the
transition relation is functional.

A \emph{play} in such a game consists of an infinite sequence of states, i.e.
$q_0 q_1 \ldots \in Q^\omega$, where $q_0 = q_I$. For a play $\pi = q_0 q_1
\ldots$, we denote by $\pi[n]$ its $(n+1)$-th state, i.e $q_{n}$. A
\emph{strategy of \adam} is a function $\stratadam \colon Q^* \to \Sigma_u$
which given a sequence of states, chooses an uncontrollable action. A
\emph{strategy of \eve} is a function $\strateve \colon Q^* \times \Sigma_u \to
\Sigma_c$ which given a sequence of states and an uncontrollable action, choses
a controllable action. For $\pi = q_0 q_1 \ldots q_n \in Q^*$, we denote by
$\last(\pi)$ the last state from $\pi$, i.e. $q_n$.  We say $\stratadam$ is a
\emph{memoryless strategy of \adam} if for any $\pi,\pi' \in Q^*$ then
$\last(\pi) = \last(\pi')$ implies $\stratadam(\pi) = \stratadam(\pi')$.
Similarly, $\strateve$ is a \emph{memoryless strategy of \eve} if for any
$\pi,\pi' \in Q^*, a_u \in \Sigma_u$, then $\last(\pi) = \last(\pi')$ implies
$\strateve(\pi,a_u) = \strateve(\pi',a_u)$.

A play $\pi$ is \emph{consistent} with a pair of strategies
$(\stratadam,\strateve)$ if for all $i \ge 0$:
\[ 
	\pi[i+1] = \delta(\pi[i], \stratadam(\pi[i]),
	                          \strateve(\pi[i], \stratadam(\pi[i]))).
\]
Given a strategy~$\strateve$ of \eve, we write $\plays(G,\strateve)$ the set of
plays that are consistent with $(\stratadam, \strateve)$ for some $\stratadam$.

A \emph{safety game} is a pair $\langle G, \mathcal{U} \rangle$ where
$\mathcal{U} \subseteq Q$ is a set of \emph{unsafe} states.  The objective of
\eve is to keep the play within the states $Q \setminus \mathcal{U}$ at all
times. We say that $\strateve$ is \emph{winning for \eve} if for any play $\pi
\in \plays(G, \strateve)$, for all $n \ge 0$, $\pi[n] \not \in \mathcal{U}$.
Otherwise, $\pi$ is \emph{winning for \adam}, and we denote by $i_\pi$ the first
turn in which a state in $\mathcal{U}$ is visited, that is $i_\pi = \min\{ i \ge
0 \st \pi[i] \in \mathcal{U}\}$.
Note that in safety games, the objective of \adam is to
\emph{reach}~$\mathcal{U}$. From the point of view of \adam, these are in fact 
\emph{reachability games}. 

In this work we study finite safety and reachability games for which it is known that
\emph{memoryless strategies} suffice for either player (see,
e.g.~\cite{gradel04}). Thus in what follows, when we speak about strategies, we
mean memoryless strategies and we take strategies for \adam and \eve to be of
the form $\stratadam : Q \to \Sigma_u$ and $\strateve : Q \times \Sigma_u \to
\Sigma_c$, respectively.

\paragraph*{Safety Games For Circuits}
We formalize the controller synthesis problem for circuits as safety games.
Given a circuit specification $\langle X_u, X_c, L, (f_l)_{l\in L}, \fbad
\rangle$, we define the game $\langle G, \mathcal{U} \rangle$ with $G=\langle
Q,q_I, \Sigma_u, \Sigma_c, \delta \rangle$, where $Q = \mathbb{B}^L$, $q_I =
0^L$ (i.e. the valuation that assings~$0$ to all~$L$),
$\Sigma_u = \mathbb{B}^{X_u}$, and $\Sigma_c = \mathbb{B}^{X_c}$. So
states (resp. actions) in~$G$ are valuations on latches (resp. inputs).
Let $q,s$ be valuations on latches. We define accordingly the transition
function as $\delta(q,\sigma_u,\sigma_c) \mapsto s$ if $s(l) =
f_l(q,\sigma_u,\sigma_c)$ for all $l \in L$.

\section{Realizability}

\paragraph*{Basic Fixpoint Algorithm}
We recall the basic fixpoint computation for solving safety games, applied here
on safety games for circuits.  Let $C = \langle X_u, X_c, L, (f_l)_{l\in
L},\fbad \rangle$ be a circuit specification and $G_C = \langle Q, q_I,
\Sigma_u, \Sigma_c, \delta, \mathcal{U} \rangle$ the associated safety game.
The set of the states from which there is no \eve's strategy to ensure the
safety objective can be computed by iterating an \emph{uncontrollable
predecessors} operator.  For any $S \subseteq Q$, the \emph{uncontrollable
predecessors} of~$S$ is defined as
\[
	\upre(S) = \{q \in Q \st \exists \sigma_u \in \Sigma_u.\ \forall \sigma_c \in
		\Sigma_c :
		\delta(q, \sigma_u, \sigma_c) \in S\}.
\]
We denote by $\upre^*(S) = \mu X. (S \cup \upre(X))$, the \emph{least fixpoint}
of the function $F : X \to S \cup \upre(X)$ in the $\mu$-calculus notation (see
\cite{ej91}). Note that $F$ is defined on the powerset lattice, which is finite.
It follows from Tarski-Knaster theorem~\cite{tarski55} that, because $F$ is
monotonic, the fixpoint exists and can be computed by iterating the application
of $F$ starting from the least value of the lattice, i.e. $\emptyset$.

The following is a well-known result about the relationship between safety games
and the $\upre$ operator. The second part of the claim follows from the
determinacy of finite safety games~\cite{gradel04}.
\begin{proposition}
\label{pro:win-lose-safe}
	Let $C$ be a circuit specification and $G_C$ the associated safety game. Then
	\begin{inparaenum}[(i)]
		\item \adam has a winning strategy in $G_C$ if and only if $q_I
			\in \upre^*(\mathcal{U})$; and
		\item \eve has a winning strategy in $G_C$ if and only if $q_I
			\not\in \upre^*(\mathcal{U})$.
	\end{inparaenum}
\end{proposition}

\paragraph*{Symbolic implementation of $\upre$}
There is a plethora of symbolic algorithms to do forward and backward state
space exploration in large systems defined succinctly,
e.g.~\cite{cbm90,cmb91,bcl91} to mention a few classic works on the topic.  The
construction of a symbolic (monolithic or partitioned) transition relation is
the first step of those algorithms. For deterministic systems, where the
transition relation is functional, a \emph{transition function vector} can be
used to represent the transitions (that is, one distinct function for each
latch).  This is known to improve the performance of state space exploration
algorithms in some systems, although this is not the case systematically;
see~\cite{cbm90,cmb91}.

We consider both monolithic and partitioned transition relations in this work.
We present 1) a version of the operators using the monolithic transition
relation $T(L,X_u,X_c,L')$ constructed once at the beginning of the algorithm,
and 2) an alternative version using only the partitioned transition relation.
Our results also confirm that the preference between the two depends on the type
of circuit (see Section~\ref{sec:experiments-discussion}).

More precisely, the monolithic transition relation is defined as
$T(L,X_u,X_c,L') = \bigwedge_{l \in L} l' \Leftrightarrow f_l(X_u,X_c,L)$, where
$L'$ represents the next step states.  $\upre(S)$ can then be computed symbolically by
the formula
\[
	\upre(S) = \exists X_u.\ \forall X_c.\ \exists L': T(L,X_u,X_c,L') \land
	S(L').
\]
Alternatively, we observe that since we have that $\forall l' \in L' : l'
\Leftrightarrow f_l(X_u,X_c,L)$ we can directly substitute all $l'$ into
$S(L')$ to obtain the desired set without using $T(L,X_u,X_c,L')$, i.e.
\[
	\upre(S) = \exists X_u.\ \forall X_c : S(L')[l' \leftarrow
	f_l(X_u,X_c,L)]_{l \in L}.
\]

\section{Abstractions of Safety Games}
\subsection{Conservative Abstractions}
Computing the fixpoint of~$\upre$ in safety games for circuits may be infeasible
due to their large state spaces. For such circuits, we consider
\emph{abstractions}, which are games with smaller state spaces on which fixpoint
computations are feasible. We follow the \emph{abstract interpretation}
framework~\cite{cc77} to build \emph{conservative abstractions}, so as to make sure
that if the abstract game can be won by \eve, then the concrete game 
can also be won by her.

Let $C$ be a circuit specification and $G_C=\langle Q, q_I, \Sigma_u, \Sigma_c,
\delta, \mathcal{U} \rangle$ its associated safety game.  Intuitively,
abstractions will be obtained by partitioning the state space of~$G_C$ and
defining transitions between the elements of the partition.  Formally, a game
$G_C^\abs = \langle Q^\abs, q_I^\abs, \Sigma_u, \Sigma_c, \Delta^\abs,
\mathcal{U}^\abs \rangle$ is a \emph{conservative abstraction of $G_C$} if 
\begin{itemize}
	\item $Q^\abs$ is a partition of $Q$;
	\item $q_I^\abs = \{q_I\} \in Q^\abs$;
	\item $(s, \sigma_u, \sigma_c, s') \in \Delta^\abs$ if $\exists q \in
		s.\ \exists q' \in s' : \delta(q, \sigma_u, \sigma_c) = q'$; and
	\item $\mathcal{U} \subseteq \bigcup_{u \in \mathcal{U}^\abs} u$.
\end{itemize}

Notice that we require the abstractions to distinguish the initial states, and
the abstract safety specification $\calU^\abs$ to cover~$\calU$.  Conservative
abstractions give more power to \adam~\cite{hmmr00}. We will show that
if \eve wins in a conservative abstraction, then it wins in the original
game.
We will refer to the states of~$G_C$ as \emph{concrete states}, and those
of~$G^\abs$ as \emph{abstract states}.

We define the \emph{concretization function} $\gamma : \pow(Q^\abs) \to
\pow(Q)$ for this abstraction, defined by 
\(
	\gamma(S^\abs) = \bigcup_{s^\abs \in S^\abs} s^\abs,
\)
which gives the set of concrete states covered by a given set of abstract
states.  The dual operation is abstraction; we define two \emph{abstraction
functions} $\overline{\alpha},
\underline{\alpha} : \pow(Q) \to \pow(Q^\abs)$ as follows:
\(
	\overline{\alpha}(S) = \{ q^\abs \in Q^\abs \st S \cap q^\abs \neq
	\emptyset\},
	\) and
	\(
	\underline{\alpha}(S) = \{ q^\abs \in Q^\abs \st S \subseteq q^\abs\}.
	\)
Intuitively, $\overline{\alpha}(S)$ is the smallest set of abstract states
covering~$S$; while $\underline{\alpha}(S)$ is the largest set of abstract
states entirely included in~$S$. The pair $(\overline{\alpha},\gamma)$ defines a
Galois connection:

\begin{lemma}[from~\cite{cc77}]
	\label{lem:galois-con}
	The pair $(\overline{\alpha},\gamma)$ is a Galois connection, that is,
	for all $s \subseteq Q$ and $t \subseteq Q^\abs$, we have that
	$\overline{\alpha}(s) \subseteq t$ if, and only if $s \subseteq
	\gamma(t)$.
\end{lemma}

The following lemma shows the relation between~$\overline{\alpha}$
and~$\underline{\alpha}$, which are, respectively, over- and
under-approximations of given sets.
\begin{lemma}
	\label{lem:under-conc-over}
	For any $S \subseteq Q$, we have $\gamma(\underline{\alpha}(S))
	\subseteq S \subseteq \gamma(\overline{\alpha}(S))$, and
	$\underline{\alpha}(S) = Q^\abs \setminus \overline{\alpha}(Q \setminus
	S)$.
\end{lemma}

\subsection{Predicate Abstraction and Localization Reduction}
In order to effectively construct abstractions from the concrete circuit safety
game we use \emph{predicate abstraction} \cite{gs97} and \emph{localization
reduction} \cite{kurshan94}.  Predicate abstraction consists in defining
abstractions by partitioning the state space by predicates, and is used e.g. in
CEGAR methods~\cite{cgjlv00}.  Localization reduction is a special case of
predicate abstraction in which predicates consist of single latches.

Consider any circuit specification $C = \langle X_u, X_c, L, (f_l)_{l\in L},
\fbad \rangle$, and the associated safety game $G = \langle Q, q_I,
\Sigma_u, \Sigma_c, \delta, \calU \rangle$.
Let $P$ be a set of boolean variables, also called \emph{predicates},
and $(f_p(L))_{p \in P}$ be a set of formulas.
We assume that there exist $p_I,p_U \in P$ such that~$f_{p_I} \equiv q_I$,
and~$f_{p_U} \equiv \calU$.
The predicates~$P$ partition the state space $Q$, i.e.
$\mathbb{B}^L = \biguplus_{v \in \mathbb{B}^P} \bigcap_{p \in P} f_p^{-1}(v(p))$.
We will consider the conservative abstraction defined on this partition.

Formally, we consider the state space~$Q^\abs = \mathbb{B}^P$.
Given $S^\abs \subseteq Q^\abs$, the concretization function is
\[
\gamma(S^\abs)(L) = S^\abs(P)[p \leftarrow f_p(L)]_{p \in P}.
\]
The abstraction functions are defined accordingly:
\begin{align*}
	\overline{\alpha}(S)(P) &= \exists L: S(L) \land (\bigwedge_{p \in P} p
	\Leftrightarrow f_p(L)),&\\
	\underline{\alpha}(S)(P) &= \exists L: \lnot (\lnot S(L) \land (\bigwedge_{p
	\in P} p \Leftrightarrow f_p(L))).& \text{from
	Lemma~\ref{lem:under-conc-over}}\\
\end{align*}

The transition relation $\Delta^\abs$ is given by $(q^\abs,
\sigma_u, \sigma_c, r^\abs) \in \Delta^\abs \Leftrightarrow
q^\abs,\sigma_u,\sigma_c,r^\abs \models T^\abs$, where
\begin{align*}
	T^\abs(P,X_u,X_c,P') &= \exists L,L': T(L,X_u,X_c,L') 
	\land (\bigwedge_{p \in P} p \Leftrightarrow (f_p(L)) 
	\land (\bigwedge_{p' \in P'} p' \Leftrightarrow (f_{p'}(L')).
\end{align*}

\subsection{Abstract uncontrollable predecessors}
\label{sec:abs-upre}
We now define the uncontrollable predecessors operators in the abstract games.
We define two operators, one yielding an over-approximation of the
usual~$\upre$, and another one yielding an under-approximation.
We let
\begin{align*}
	\upreover(S^\abs) &= \{q^\abs \st \exists \sigma_u.\ \forall \sigma_c.\
	\exists r^\abs : (q^\abs,\sigma_u,\sigma_c,r^\abs) \in \Delta^\abs
	\text{ and } r^\abs \in S^\abs \},\\
	\upreunder(S^\abs) &= \{q^\abs \st \exists \sigma_u.\ \forall \sigma_c.\
	\forall r^a : (q^\abs,\sigma_u,\sigma_c,r^\abs) \in \Delta^\abs
	\text{ implies } r^\abs \in S^\abs \}.
\end{align*}

Given a formula $S^\abs(P)$, representing a set of states from $Q^\abs$, the
operators can be easily computed symbolically as follows:
\begin{align*}
	\upreover(S^\abs) &= \exists X_u.\ \forall X_c.\ \exists P' :
		T^\abs(P,X_u,X_c,P') \land S^\abs(P'),\\
	\upreunder(S^\abs) &= \exists X_u.\ \forall X_c.\ \forall P' :
		T^\abs(P,X_u,X_c,P') \Rightarrow S^\abs(P').
\end{align*}

The following lemma shows the relation between the abstract uncontrollable
predecessor operator and the concrete one.
\begin{lemma}\label{pro:abs-vals}
	For any set $S^\abs \subseteq Q^\abs$,
	\(
		\gamma(\upreunder(S^\abs)) \subseteq \upre(\gamma(S^\abs))
		\subseteq \gamma(\upreover(S^\abs)).
		\)
\end{lemma}
\begin{proof}
	We show the inequalities hold from left to right.  Let $q^\abs \in
	\upreunder(S^\abs)$. There is $\sigma_u$ such that for all $\sigma_c$,
	for any state~$r^\abs$, $(q^\abs,\sigma_u,\sigma_c,r^\abs) \in
	\Delta^\abs$ implies $r^\abs \in S^\abs$. So by construction of
	$\Delta^\abs$, for any state $q$ in
	$\gamma(q^\abs)$, for all $\sigma_c$, $\delta(q,\sigma_u,\sigma_c) \in
	\gamma(S^\abs)$. Hence $q \in \upre(\gamma(S^\abs))$ and
	$\gamma(q^\abs) \subseteq \upre(\gamma(S^\abs))$. Therefore
	$\gamma(\upreunder(S^\abs)) \subseteq \upre(\gamma(S^\abs))$.
	
	Let $q \in \upre(\gamma(S^\abs))$, then there is $\sigma_u$ such that
	for all $\sigma_c$, $\delta(q,\sigma_u,\sigma_c) \in \gamma(S^\abs)$. By
	definition of $\Delta^\abs$, for all $\sigma_c$
	$(\overline{\alpha}(q),\sigma_u,\sigma_c,
	\overline{\alpha}(\delta(q,\sigma_u,\sigma_c))) \in \Delta^\abs$. Since
	$\overline{\alpha}(\delta(q, \sigma_u,\sigma_c)) \in S^\abs$ for all
	$\sigma_c$, we have $\overline{\alpha}(q) \in \upreover(S^\abs)$ and $q
	\in \gamma(\upreover(S^\abs))$. Hence $\upre(\gamma(S^\abs)) \subseteq
	\gamma(\upreover(S^\abs))$.

\end{proof}

Lemma~\ref{pro:abs-vals} implies, by induction, the following.

\begin{lemma}
\label{lem:overappU}
	$\gamma(\upreunder^*(\calU^\abs)) \subseteq \upre^*(\calU)) \subseteq
	\gamma(\upreover^*(\calU^\abs))$.
\end{lemma}
This yields the following Theorem.
\begin{theorem}
\label{thm:conservative-abstraction}
	Let~$C$ be a circuit specification, $G_C$ be its associated safety game
	and $G_C^\abs$ a conservative abstraction of it. If \eve wins~$G^a_C$
	then she also wins~$G_C$.
\end{theorem}

\subsection{Optimizations}
\label{sec:optimized-ops}

Because we consider conservative abstractions, if \adam wins in the abstract
game, one cannot conclude unrealizability right away. However, we can still use
the information gathered during the computation of the abstract uncontrollable
predecessors $\upreover$. In fact, we will show that the states and actions that
witness the uncontrollable predecessors in each iteration of~$\upreover$ define
a \emph{set} of strategies which contains any concrete winning strategy for
\adam (Prop.~\ref{pro:all-unsafe-strats} below).  We then use this information
to restrict future $\upre$ operations to these strategies.

\paragraph{Quasi-strategies of \adam}
Formally, a \emph{quasi-strategy of \adam} in the conservative abstraction
$G^\abs$ of a game $G$ is a function $\qstratadam : Q^\abs \to \pow(\Sigma_u)$
which maps any abstract state to a \emph{set} of uncontrollable actions. Thus, a
\emph{quasi-strategy} can be seen as a non-deterministic strategy defined on a
subset of states (in fact, $\qstratadam$ can map some states to the empty set).
We denote by $\cqstratadam$ the concrete quasi-strategy in~$G$ given by
\(
	\cqstratadam(q) \mapsto \qstratadam(\overline{\alpha}(q))
\)
for any~$q \in Q$.

Let $S^\abs \subseteq Q^\abs$, and $W = \upreover^*(S^\abs)$. The set $W$
describes a quasi-strategy $\qstratadam_W$, defined by
\(
	\qstratadam_W(q) \mapsto \{\sigma_u \in \Sigma_u \st \forall \sigma_c.\
		\exists s \in W : (q, \sigma_u, \sigma_c, s) \in
		\Delta^\abs \}, 
\)
for any~$q \in Q^\abs$. This quasi-strategy corresponds to the set of
uncontrollable actions \adam can play from states in $W$ to stay within $W$.
Note that not all strategies respecting~$\qstratadam_W$ are winning for \adam; although
all winning strategies for \adam choose actions prescribed by
$\qstratadam_W$.

\begin{proposition}
	\label{pro:all-unsafe-strats}
	Let $G^\abs$ be a conservative abstraction of a game $G$, $\qstratadam$
	be the quasi-strategy for \adam defined by
	$\upreover^*(\mathcal{U}^\abs)$ and $\stratadam$ a strategy for \adam in
	$G$. If $\stratadam$ is a winning strategy for \adam in $G$, then
	$\forall \pi \in \plays(G, \stratadam).\ \forall i < i_\pi:
	\stratadam(\pi[i]) \in \gamma(\qstratadam)(\pi[i])$.
\end{proposition}
\begin{proof}
	Let $S^\abs \subseteq Q^\abs$ and $\stratopadam^i(S^\abs) = \{
	(q, \sigma_u) \in Q^\abs \times \Sigma_u \st \forall \Sigma_c.\
	\exists r : (q,\sigma_u,\sigma_c) \in \Delta^\abs$ and $r \in
	\upreover^i(S^\abs)\}$. Clearly $\stratopadam^{i-1}(S^\abs) \sqsubseteq
	\qstratadam$.
	
	Let $\iota = \max\{ i_\pi \st \pi \in \plays(G,\stratadam)\}$, which is finite
	since $\stratadam$ is memoryless and winning for \adam.

	Consider any play~$\pi \in \plays(G, \stratadam)$. We show that
	\(
		\stratadam(\pi[\iota - j]) \in \gamma(\stratopadam^{j}(\mathcal{U}^\abs)).
	\)

	By definition of~$\iota$, for all $\sigma_c \in \Sigma_c$,
	$\delta(\psi[\iota - 1], \stratadam(\psi[\iota - 1]), \sigma_c) \in
	\mathcal{U}$.  By construction of $\Delta^\abs$, we have that for all
	$\sigma_c \in \Sigma_c$ there exists $r^\abs$ such that
	$(\overline{\alpha}(\psi[\iota - 1]),\stratadam(\psi[\iota -
	1]),\sigma_c,r^\abs) \in \Delta^\abs$ where $r^\abs \in
	\overline{\alpha}(\mathcal{U})$. This implies that
	$(\overline{\alpha}(\psi[\iota - 1]), \stratadam(\psi[\iota - 1])) \in
	\stratopadam(\overline{\alpha}(\mathcal{U}))$. Note that from the
	definition of $\mathcal{U}^\abs$ we get that
	$\overline{\alpha}(\mathcal{U}) \subseteq \mathcal{U}^\abs$ and that
	since $\stratopadam$ is monotone $ (\overline{\alpha}(\psi[\iota - 1]),
	\stratadam(\psi[\iota - 1])) \in \stratopadam(\mathcal{U}^\abs)$. Thus
	$\stratadam(\psi[\iota - 1]) \in \gamma(\qstratadam)(\psi[\iota - 1])$. 
	
	Consider now $2\leq j \leq \iota$.  For all $\sigma_c \in \Sigma_c$,
	$\delta(\psi[\iota - j], \stratadam(\psi[\iota - j]), \sigma_c) \in
	\upre^{j-1}(\calU)$.  It follows that for any~$\sigma_c \in \Sigma_c$,
	there exists $r^\abs \in \overline{\alpha}(\upre^{j-1}(\calU))$ with
	$(\overline{\alpha}(\psi[\iota-j]), \stratadam(\psi[\iota - j]),
	\sigma_c, r^\abs) \in \Delta^\abs$.  We have $r^\abs \in
	\overline{\alpha}(\upre^{j-1}(\gamma(\calU^\abs)))$ by monotonicity
	of~$\upre$, and by Lemma~\ref{lem:overappU}, we get $r^\abs \in
	\upreover^{j-1}(\calU^\abs)$. Hence, $(\overline{\alpha}(\psi[\iota-j]),
	\stratadam(\psi[\iota-j])) \in \stratopadam(\upreover^{j-1}(\calU^\abs))
	\subseteq \stratopadam^j(\calU^\abs)$.
\end{proof}

\paragraph{Guiding $\upre$ using $\qstratadam_W$}
For convenience, let $\qstratadam = \qstratadam_W$.  We define the
concrete $\upre$ operator restricted to the quasi-strategy
$\cqstratadam$ as follows.
\[
	\upre_\cqstratadam(S) = \{q \in Q \st \exists \sigma_u
		\in \cqstratadam(q).\
		\forall \sigma_c \in \Sigma_c:
		\delta (q^\abs,\sigma_u,\sigma_c) \in S\}.
\]

$\upre_\cqstratadam(S)$ yields the set states from which \adam can force to
reach states in $S$ by using actions compatible with the given quasi-strategy.
Proposition~\ref{pro:all-unsafe-strats} implies that because the quasi-strategy
was extracted from the abstract uncontrollable predecessors fixpoint, this
restriction is no loss of generality. Indeed, if \adam has a winning strategy it
is included in the quasi-strategy. It follows that if the abstract game is
winning for \eve, then this will be detected by $\upre$ restricted
to~$\gamma(\qstratadam)$.
\begin{theorem}
	Let~$G^\abs$ be a conservative abstraction of a game~$G$,
	and~$\qstratadam$ be the quasi-strategy for \adam defined by
	$\upreover^*(\calU^\abs)$.  $q_I \not \in
	\upre^*_{\gamma(\qstratadam)}(\calU)$ if and only if $q_I \not \in
	\upre^*(\calU)$.
\end{theorem}
\begin{proof}
	Observe that since $\upre_{\gamma(\qstratadam)}$ is a restricted version
	of $\upre$, we have that $\upre^i_{\gamma(\qstratadam)}(S) \subseteq
	\upre^i(S)$ for any $S \subseteq Q$ and all $i \ge 0$. Thus, if $q_I
	\not \in \upre^*(\calU)$ then $q_I \not \in
	\upre^*_{\gamma(\qstratadam)}(\calU)$.

	For the other direction recall that from
	Proposition~\ref{pro:win-lose-safe} we have that $q_I \in
	\upre^*(\calU)$ if and only if \adam has a winning strategy in $G$.
	Assume $q_I \in \upre^*(\calU)$ and that $\stratadam$ is a winning
	strategy for \adam in $G$. By Proposition~\ref{pro:all-unsafe-strats} we
	get that
	\[
		\bigcup_{\substack{\pi \in \plays(G,\stratadam)\\ 0 \le j \le
		i_\pi}} \pi[j] \subseteq \upre^*_{\gamma(\qstratadam)}(\calU).
	\]
	In particular, this implies that $q_I \in
	\upre^*_{\gamma(\qstratadam)}(\calU)$, which concludes the proof.
\end{proof}

\paragraph{Reachable states under $\qstratadam_W$}
As a second optimization, we restrict the exploration of both the concrete and
abstract state spaces to the set of states which are reachable from the initial
state when \adam plays according to a winning strategy. This will allow us to
prune the search space. As we will show, the set of states that are winning for \adam
but not reachable from the initial state, or those states reached by strategies
losing for \adam can be safely ignored.

Let $\post(S, \stratadam)$ be the set of states reachable from $s \in Q$ if
\adam plays according to $\stratadam$. We now formally define
\[
	\calR(G) = \bigcup_{\substack{\stratadam \text{ winning for \shortadam}\\ \pi \in
	\plays(G,\stratadam)\\ i \ge 0}} \pi[i].
\]
Note that $\calR(G)$ is empty if the circuit is controllable. We will omit~$G$
from~$\calR(G)$ when it is clear from the context. 
Ideally, we would like to restrict our computations to~$\calR(G)$. However,
computing~$\calR(G)$ is clearly as difficult as solving realizability
of the safety game $G$, so we will rather
consider over-approximations of this set computed on conservative abstractions
of~$G^\abs$.
For any $S^\abs \subseteq Q^\abs$, and $\qstratadam$ a
quasi-strategy for $\adam$ in $G^\abs$, the \emph{possible successors
under $\qstratadam$} are defined as follows.
\begin{align*}
	\post(S^\abs, \qstratadam) = \{r^\abs \in Q^\abs \st
	\exists q^\abs \in S^\abs.\ \exists \sigma_u. \in
	\qstratadam(q^\abs).\ \exists \sigma_c \in \Sigma_c: 
	(q^\abs, \sigma_u, \sigma_c, r^\abs) \in \Delta^\abs\}.
\end{align*}
Note that the $\post$ operator can be computed symbolically as follows.
\[
	\post(S^\abs, \qstratadam) = \exists P .\ \exists X_u.\
	\exists X_c : T^\abs(P, X_u, X_c, P') \land S^\abs(P) \land
	\qstratadam(P, X_u).
\]
Let $G^\abs$ be a conservative abstraction of a game $G$ and $\qstratadam_W$ the
quasi-strategy defined by $\upreover^*(\mathcal{U}^\abs)$.
Now, Prop.~\ref{pro:all-unsafe-strats} implies the following result.
\begin{proposition}
	Let $G^\abs$ be a conservative abstraction of a game $G$ and
	$\qstratadam_W$ the quasi-strategy defined by
	$\upreover^*(\calU^\abs)$. Then $\calR(G) \subseteq \gamma(\post^*(q_I^\abs,
	\qstratadam_W))$.
\end{proposition}

Now, the first purpose of defining over-approximations of~$\calR(G)$ is to
restrict the fixpoint computations on the abstract game to these states, so that
the considered sets of states are smaller. This will, hopefully, lead to smaller
BDDs. We define the $\upreover$ fixpoint computation restricted to
over-approximations of~$\calR$.

\begin{theorem}
	\label{thm:fp-over-rest}
	Let $G^\abs$ be a conservative abstraction of a safety game $G$,
	and let~$R^\abs \subseteq Q^\abs$ with $\calR \subseteq \gamma(R^\abs)$. 
	Then $\gamma(\mu X.\ \mathcal{U}^\abs \cup \upreover(X)) \cap \calR =
	\gamma(\mu X.\ (\mathcal{U}^\abs \cup \upreover(X)) \cap R^\abs) \cap \calR$.
\end{theorem}

The same idea can be applied to the post operator.
\begin{theorem}
	\label{thm:post-Ra}
	Let $G^\abs$ be a conservative abstraction of a safety game $G$,
	and let~$R^\abs \subseteq Q^\abs$ with $\calR \subseteq \gamma(R^\abs)$. 
	Then $\gamma(\mu X.\ \{q_I^\abs\} \cup \post(X)) \cap \calR =
	\gamma((\mu X.\ \{q_I^\abs\} \cup \post(X,\qstratadam)) \cap R^\abs ) \cap
	\calR$.
\end{theorem}

\paragraph{Using Abstract Partitioned Transition Relation}
As mentioned earlier, in some circuits, one can improve performance by using
only a partitioned transition relation and avoiding the computation of the
monolithic transition relation. In this paragraph, we explain how this can be
achieved and combined with the reachability analysis in abstract games.

Note that partitioning the transition relation works well in instances in which
the transition relation is large (i.e. the size of the BDD needed to represent
$T^\abs$ is large) but the fixpoint is reached in a small number of steps.  On
the contrary, if too many iterations are needed to obtain the fixpoint, then it
is often more efficient to construct $T^\abs$ once and use it to compute the
operators $\upreover$ and $\upreunder$, as the cost will be amortized in the long run. These
observations are illustrated in the section on experiments.
\medskip

Let $\psi_p(L,X_u,X_c) = f_p(L')[l' \leftarrow f_l(X_u,X_c,L)]_{l \in L}$.  Then
the $\upreover,\upreunder$ operators can be computed as shown below.
\begin{lemma}
	\label{lemma:partitioned-upre}
	For any~$G^\abs$, 
	\begin{align*}
		\upreover(S^\abs) &= \exists X_u.\ \forall X_c : \overline{\alpha}(S^\abs(P')[p'
			\leftarrow \psi_p(X_u,X_c,L)]_{p \in P}) \\
		\upreunder(S^\abs) &= \lnot (\forall X_u.\ \exists X_c :
			\overline{\alpha}(\lnot S^\abs(P')[p' \leftarrow \psi_p(X_u,X_c,L)]_{p \in
			P})).
	\end{align*}
\end{lemma}

We also present an operator yielding an over-approximation of the set of
reachable states which can be computed with partitioned transition relations.
Let $S^\abs \in Q^\abs$ and $\qstratadam$ be a quasi-strategy for \adam in
$G^\abs$. 
\[
	\overline{\post}(S^\abs, \qstratadam) = \exists L, X_u : (S^\abs(P) \land
	\qstratadam(P,X_u))[p \leftarrow f_p(L)]_{p \in P} \land
	\bigwedge_{p \in P} \exists X_c: p' \Leftrightarrow \psi_p(X_u,X_c,L).
\]
Note that $\overline{\post}$ is defined, from~$\post$, simply pushing the
quantification over~$X_c$ inside. In fact, the exact definition of~$\post$
contains the transition relation~$T^\abs$, which we want to avoid computing.

The following lemma shows that this yields over-approximations.
\begin{lemma}
	\label{lemma:partitioned-reach}
	The set of abstract states reachable from $S^\abs$ in one step
	if \adam plays according to $\qstratadam$ is contained in
	$\overline{\post}(S^\abs, \qstratadam)$. That is, $\post(S^\abs, \qstratadam)
	\subseteq \overline{\post}(S^\abs, \qstratadam)$.
\end{lemma}
Note that one could also push the quantification over $X_u$ inside the
conjunction in order to obtain coarser over-approximations. However, this
alternative definition was not faster to compute, nor did it improve overall
performance in our experiments.

\section{Yet another CEGAR algorithm}\label{sec:cegar}
We present a CEGAR-based synthesis algorithm, given in
Algorithm~\ref{alg:abs_synth}, based on a combination of ideas introduced
in~\cite{ap07} and~\cite{hjm03}. The algorithm constructs abstractions using
-- initially -- three predicates, namely, $p_I$ describing the initial state,
$p_U$ an under-approximation of the losing states, and $p_R$ an
over-approximation of the states reachable from the initial state under winning
strategies of \adam.  The algorithm further refines the abstraction by
localization reduction.  In fact, the initial abstraction consists of the
conservative abstraction defined by these three predicates, and at each
refinement loop, some latch is made ``visible'', that is, added as a predicate.

We give an informal description of the algorithm. Given a conservative
abstraction, the algorithm first computes $W_u$ at line~$1$, the fixpoint
of~$\upreunder$, restricted to~$R^\abs$ which overapproximates~$\calR(G)$.
If the initial
state belongs to~$W_u$, then by Lemma \ref{lem:overappU}, \eve has no winning
strategy (line~$3$). Otherwise, in the while loop of line~$7$, we compute~$W_o$,
the fixpoint for~$\upreover$ restricted to $R^\abs$ -- which is an
over-approximation on reachable states under winning strategies of \adam.  In
this case, if the initial state does \emph{not} belong to~$W_o$, then nor does
it belong to the fixpoint of~$\upre$
and we conclude that the circuit is controllable. Otherwise, we recompute the
fixpoint for~$\upreover$ by decreasing~$R^\abs$: we first compute, at line~$13$,
the quasi-strategy for \adam allowing her to stay inside~$W_o$, then
restrict~$R^\abs$, at line~$14$, to states that are reachable under this
quasi-strategy.  These restrictions are justified since any winning strategy for
\adam is compatible with these (see
Proposition~\ref{pro:all-unsafe-strats}).  If we were not able to conclude, then
the abstraction is too coarse and needs to be refined.  At line~$17$, we compute
the concrete $\upre$ of~$W_u$ restricted to the quasi-strategy and to
$R^\abs$. If it turns out that $W_u$ was already a fixpoint for~$\upre$, then
we know that the circuit is controllable (line~$19$) since $W_u$ does not
contain the initial state. Otherwise, we refine the abstraction by making a
latch visible, but also increasing $\calU^\abs$ using the information computed
with $\upreunder$.  The refinement step is given by the \texttt{refine} function
described in Algorithm~\ref{alg:refine}.

The algorithm is initially called with the abstraction given by the three
predicates $\{p_I,p_U,p_R\}$ defined by $p_I \equiv \{q_I\}$, $p_U \equiv \calU$,
and $p_R \equiv Q$.

\begin{algorithm}
	\small
	\KwData{Safety game $G = \langle Q, q_I, \Sigma_u, \Sigma_c, \delta,
		\mathcal{U} \rangle$, abstraction $G^\abs =
		\langle Q^\abs, q_I^\abs, \Sigma_u, \Sigma_c, \Delta^\abs,
		\mathcal{U}^\abs \rangle$ and $R^\abs \supseteq \calR$.}
	$W_u$ := $\mu X .\ (\mathcal{U}^\abs \cup
		\upreunder(X)) \cap R^\abs$\;
	\If{$q_I^\abs \in W_u$}
	{
		\Return not controllable\;
	}
		$prev$ := $\emptyset$\;
		\While{$R^\abs \neq prev$}
		{
			$prev$ := $R^\abs$\;
			$W_o$ := $\mu X .\ (W_u \cup
				\upreover(X)) \cap R^\abs$\;
			\If{$q_I^\abs \not\in W_o$}
			{
				\Return controllable\;
			}
				$\qstratadam$ := quasi-strategy defined by $(W_o)$\;
				$R^\abs$ := $\mu X .\ (q^\abs_I \cup \post(X,
					\qstratadam)) \cap R^\abs$\;
		}
		$W'_u$ := $(\upre_\cqstratadam(\gamma(W_u)))
			\cap \gamma(R^\abs)$\;
		\If{$W'_u \subseteq \gamma(W_u)$}
		{
			\Return controllable\;
		}
			$Q^\abs_2$ := \texttt{refine}$(Q^\abs, W'_u \cup
				\gamma(W_u), \gamma(R^\abs))$\tcp*[r]{$\underline{\alpha}_2,\overline{\alpha}_2$ are the associated abstraction operators}
			$\mathcal{U}^\abs_2$ := $\underline{\alpha}_2(W'_u \cup
				\gamma(W_u))$\;
			\Return \texttt{abs\_synth}$(G, G^\abs_2,
			\overline{\alpha}_2(\gamma(R^\abs)))$\;
\caption{\texttt{abs\_synth}$(G, G^\abs, R^\abs)$}
\label{alg:abs_synth}
\end{algorithm}

Refinement is achieved symbollicaly by adding a new predicate to our predicate
set $P$. Besides having $W'_u \cup \gamma(W_u)$ and $\gamma(R^\abs)$ replace the
previous $p_U$ and $p_R$, respectively, we also make a new latch ``visible''.
Latches that depend on the value of other visible latches are given priority by
Algorithm~\ref{alg:refine}.

\begin{algorithm}
	\small
	\KwData{Predicate set $P=\{p_I,p_U,p_R,l_{\alpha_1},\ldots,l_{\alpha_m}\}$,
	and sets $U'(L)$ and $R'(L)$}
	$P'$ := $P \setminus \{p_U,p_R,p_I\}$\;
	$interesting$ := $\{m \in L \setminus P' \st m \not\Rightarrow U \text{ and } \lnot
		m \not\Rightarrow U\}$\;
		$useful$ := $\{m \in interesting \st \supp(f_m) \cap
		P' \neq \emptyset \}$\;
	\uIf{$useful \neq \emptyset$}
	{
		$e$ := an element from $useful$\;
	}
	\Else
	{
		$e$ := an element from $interesting$\;
	}
	\Return $P' \cup \{e, U'(L), R'(L), p_I\}$\;
\caption{\texttt{refine}$(P, U(L), R(L))$}
\label{alg:refine}
\end{algorithm}

\begin{theorem}
	\label{thm:correctness}
	Let $G$ be a safety game, $G^\abs$ a conservative abstraction of it and
	$R^\abs \supseteq \post^*(q_I^\abs, \qstratadam)$ where $\qstratadam$ is
	the quasi-strategy defined by $\upreover^*(\mathcal{U}^\abs)$. If
	Algorithm~\ref{alg:abs_synth} returns \emph{controllable} for $(G,
	G^\abs, R^\abs)$ then \eve has a winning strategy in $G$; if it returns
	\emph{not controllable} then \adam has a winning strategy in~$G$.
	Moreover, the algorithms always terminates.
\end{theorem}

To prove the correctness of the algorithm, we first show the following
invariants.
\begin{lemma}
	\label{lem:invariants-algo}
	The following invariants hold: 
	\begin{align}
		\label{invar:Ra}
		\calR \subseteq \gamma(R^\abs),\\
		\label{invar:Ua}
		\calU \cap \calR \subseteq \gamma(\calU^\abs) \subseteq \upre^*(\calU),\\
		\label{invar:Wu}
		\calU \cap \calR \subseteq \gamma(W_u) \subseteq \upre^*(\calU),\\
		\label{invar:Wo}
		\calR \subseteq \gamma(W_o).
	\end{align}
\end{lemma}
\begin{proof}
	We prove these invariants by induction on the number of recursive calls.
	Initially, $R^\abs \equiv Q^\abs$ which satisfies \eqref{invar:Ra} by
	Proposition~\ref{pro:all-unsafe-strats}, and~\eqref{invar:Ua} is satisfied
	since~$\gamma(\calU^\abs) = \calU$.
	Consider any recursive call of the algorithm, and assume that
	\eqref{invar:Ra} and~\eqref{invar:Ua} hold at line~$1$.

	$W_u$ is defined at line~$1$. Let us show that $\calU \cap
	\calR \subseteq \gamma(W_u)$. This holds at any iteration of the
	fixpoint defining~$W_u$. In fact, we have $\calU \cap \calR \subseteq
	\gamma(\calU^\abs)$, and $\calR \subseteq \gamma(R^\abs)$, so
	any iterate contains $\calU^\abs \cap R^\abs$. The result follows
	since $\calU \cap \calR \subseteq \gamma(\calU^\abs) \cap \gamma(R^\abs)
	\subseteq \gamma(\calU^\abs \cap R^\abs)$. To show the right hand side
	inequality, it suffices to observe that $\gamma(\mu X. (\calU^\abs \cup
	\upreunder(X))) \subseteq \upre^*(\calU)$, which holds since
	$\gamma(\calU^\abs) \subseteq \upre^*(\calU)$. The inequality then follows by
	monotonicity.

	Now we analyze the while loop of line 7 to prove~\eqref{invar:Wo}
	and \eqref{invar:Ra} hold.
	Let us define $W_o' = \mu X. (W_u \cup \upreover(X))$.
	Note that we just showed $\calU \cap \calR \subseteq \gamma(W_u)$
	so $\upre^*(\calU \cap \calR) \subseteq \gamma(\upreover^*(\calU \cap \calR))
	\subseteq \gamma(W_o')$.
	But~$\calR \subseteq \upre^*(\calU \cap \calR)$ by the definition of~$\calR$.
	Moreover, by Theorem~\ref{thm:fp-over-rest},
	$\gamma(W_o) \cap \calR = \gamma(W_o') \cap \calR$. It follows that $\calR
	\subseteq \gamma(W_o)$.
	We proved the invariant for arbitrary~$R^\abs$ satisfying $\calR \subseteq
	\gamma(R^\abs)$.

	We now prove invariant~\eqref{invar:Ra} on this while loop.
	In fact, because~$\calR \subseteq \gamma(W_o)$, the strategy $\qstratadam$
	defined on line~$13$ contains all winning strategies for \adam, in the
	sense of Prop.~\ref{pro:all-unsafe-strats}.
	Now, if we denote $R' = \mu X. (q_I^\abs \cup \post(X,\qstratadam))$, then
	$\calR \subseteq \gamma(R')$ by Prop.~\ref{pro:all-unsafe-strats}.
	By Theorem~\ref{thm:post-Ra}, it follows that~$\calR \subseteq \gamma(R^\abs)$.

	It remains to show that the invariants hold on the recursive call at line 23.
	Variable $R^a$ is not modified, so we need to show \eqref{invar:Ua},
	that is, $\calU \cap \calR \subseteq \gamma(\calU^\abs_2) \subseteq
	\upre^*(\calU)$.
	By the definition of~$W_u'$ at line~$17$, we have
	that $W_u' \subseteq \upre_{\gamma(\qstratadam)}(\gamma(W_u))$, and
	since~$\gamma(W_u) \subseteq \upre^*(\calU)$, we get that
	$W_u' \subseteq \upre^*(\calU)$. Thus, $W_u'\cup \gamma(W_u) \subseteq
	\upre^*(\calU)$,  and $\gamma_2(\calU^\abs_2) = \gamma_2(\underline{\alpha}_2(W_i'
	\cup \gamma_2(W_u))) \subseteq \gamma_2(\underline{\alpha}_2(\upre^*(\calU)))
	\subseteq \upre^*(\calU)$.
	To show that other inclusion, it suffices to note that $\calU \cap \calR
	\subseteq \gamma_2(W_u) \subseteq \gamma_2(\underline{\alpha}_2(\gamma_2(W_u))
	\subseteq \gamma_2(\calU^\abs_2)$.
\end{proof}

\begin{proof}[Proof of Theorem~\ref{thm:correctness}]
	Assume that the algorithm answers not controllable, on line~$3$.
	By \eqref{invar:Wu}, we have $\gamma(W_u) \subseteq \upre^*(\calU)$
	so $q_I^\abs \in W_u$ implies $\{q_I\} = \gamma(q_I^\abs) \subseteq
	\upre^*(\calU)$, which means that \adam has a winning strategy.

	Assume that the algorithm answers controllable on line~$10$.
	By \eqref{invar:Wo}, we have $\calR \subseteq \gamma(W_o)$,
	so $q_I^\abs \not \in W_o$ means that $q_I \not \in \calR$, so \eve has
	a winning strategy.

	Last, assume that the algorithm returns controllable on line~$18$.
	We have that $q_I \in \upre^*(\calU)$ iff $q_I \in \upre^*(\gamma(W_u))$
	iff $q_I \in \upre^*_{\gamma(\qstratadam)}(\gamma(W_u))$
	iff $q_I \in \upre^*_{\gamma(\qstratadam)}(\gamma(W_u))\cap \calR$.
	The test of line~$17$ means that $\gamma(W_u)$ is already a fixpoint of the
	latter equation. Moreover, we know that~$q_I^\abs \not \in W_u$ by line~$2$.
	It follows that $q_I \not \in \upre^*(\calU)$, and the returned result is
	correct.

	Now, termination follows from the fact that at each recursive call, a new
	latch is made visible, so at most after $|L|$ iterations, we obtain the
	concrete game. In this case, $\upreover=\upreunder=\upre$, thus $W_u$ and~$W_o$ 
	are complementary inside~$R^\abs$. So the algorithm will either output not controllable
	on line~$3$, or controllable on line~$10$.
\end{proof}

\section{Strategy Synthesis}
\label{sec:synthesis}
The first step of the strategy synthesis is to obtain the \emph{winning region}
for \eve, that is, the set~$W$ of all winning states for~\eve.
With the basic fixpoint algorithm -- without
abstractions, the algorithm computes $\upre^*(\calU)$ to decide that the
circuit is controllable, so the complement of this set is the winning region.
When Algorithm~\ref{alg:abs_synth} determines the controllability of a given
game, we compute a winning region as follows.
We have, by Invariant~\eqref{invar:Wu}, that
$\gamma(W_u) \subseteq \upre^*(\calU)$, so
$\gamma(W_u)^c$ is an over-approximation of the winning region.
Then $\cpre^*(\gamma(W_u)^c)$ gives the winning region,
where $\cpre^*(X) = \nu Y. (X \cap \cpre(Y))$, 
and $\cpre(X) = \{ q \mid \forall \sigma_u\in \Sigma_u, \exists \sigma_c \in
\Sigma_c, \delta(q,\sigma_u,\sigma_c) \in X\}$.

Let~$\calS$ denote such a winning region.
As a first step, it is easy to derive a quasi-strategy for \eve from~$\calS$:
We define $\lambda$ as $\lambda(q,\sigma_u) = \{ \sigma_c \in \Sigma_c \mid
\delta(q,\sigma_u,\sigma_c) \in \calS\}$ for all~$q \in \calS$, and~$\sigma_u
\in \Sigma_u$, and arbitrarily on other states.
We denote by $\calR(\lambda)$ the set of states reachable from $q_I$ when \eve
plays any strategy compatible with the quasi-strategy $\lambda$.
It is clear that $\calR(\lambda) \cap \calU = \emptyset$. In other terms, any
strategy compatible with~$\lambda$ is winning for \eve from states~$\calS$.

We are interested in synthesizing a circuit implementing a winning strategy. 
However, the quasi-strategy we just constructed is non-deterministic, so it
cannot be directly mapped as a circuit. We are going to extract a
deterministic strategy from~$\lambda$, and show how the implementing circuit can
be produced.

\begin{algorithm}
	\small
	\KwData{Winning quasi-strategy~$\lambda(L,X_u,X_c)$, and set~$\calR(\lambda)$}
	\KwResult{A circuit for each~$\sigma_c \in \Sigma_c$, implementing a
	strategy compatible with~$\lambda$}
	\For{$x \in X_c$}
	{
		$f(L,X_u,x)$ := $\exists X_c \setminus \{x\} : \lambda(L,X_u,X_c)$\;
		$f_x(L,X_u)$ := $f(L,X_u,x)[x \leftarrow 1]$\;
		$f_{\overline{x}}(L,X_u)$ := $f(L,X_u,x)[x \leftarrow 0]$\;
		$care(L,X_u)$ := $R(L) \land (\lnot f_x(L,X_u) \lor \lnot
			f_{\overline{x}}(L,X_u))$\;
		\tcc{could also be $(\lnot
			f_{\overline{x}}(L,X_u)) \Downarrow \text{care}(L,X_u)$}
			$g_x$ := $f_x(L,X_u) \Downarrow \text{care}(L,X_u)$\;
		$\lambda$ := $\lambda \land (x \Leftrightarrow g_x(L,X_u))$\;
	}
	\Return $(g_x)_{x \in X_c}$\;
\caption{\texttt{det\_strat}$(\lambda(L,X_u,X_c), R(L))$}
\label{alg:det_strat}
\end{algorithm}

The idea of Algorithm~\ref{alg:det_strat} is to extract functions for
each~$x\in X_c$ incrementally, so that at the $i$-th iteration, 
the quasi-strategy yields a unique value for the first~$i$ controllable inputs.
To extract deterministic strategies, we use the \emph{restrict} operation
implemented in most BDD packages (see~\cite{sw96}). Given two formulas $f(X)$ and $D(X)$, the
\emph{restriction of~$f(X)$ to~$D(X)$} is defined by $(f \Downarrow D)(X)$, and
has the following property.
\begin{lemma}[from~\cite{cmb91}]
	For any two formulas $f(X)$, $D(X)$, $(f \Downarrow D)(X)$ is a set that agrees
	with~$f(X)$ on the domain~$D(X)$. In other terms,
	\(
		\forall X. D(X) \Rightarrow \big((f \Downarrow D)(X) \Leftrightarrow f(X)
	\big).
	\)
\end{lemma}\GP{changed definition to lemma, do you agree?}
This operation is useful when one needs an \emph{arbitrary} set which complies
with~$D(X)$ since the size of the BDD representing~$(f \Downarrow D)(X)$ is
guaranteed to be not larger than~$f(X)$, and is often smaller.

We will use this operation to extract functions as follows. In
Algorithm~\ref{alg:det_strat}, given $x \in X_c$, we identify the set $\textit{care}(L,X_u)$ on
which the strategy being constructed yields a unique value for~$x$
given~$L,X_u$, while we know that outside this set~$x$ could get any value.
We then define the strategy for~$x$ on this set, and (arbitrarily) extend to the whole domain by the restrict 
operation. Note that the restrict operation is an optimization; we could instead
simply set $g_x := f_x(L,X_u)$ on line~$6$.

\begin{theorem}
	Let $G$ be a safety game, $R$ a winning region for~\eve, and $\lambda$
	quasi-strategy of \eve winning from~$R$.  Then the strategy $\lambda'$
	returned by Algorithm~\ref{alg:det_strat} is winning for \eve.
\end{theorem}
\begin{proof}
	Let $x_1, x_2, \ldots$ be the ordered sequence of controllable inputs
	taken by the loop. Note that the function~$g_{x_i}$ is defined on
	iteration~$i$. 
	Let us denote by~$\lambda_0$ the quasi-strategy given in input.
	
	We show that the following invariant holds at
	the beginning of iteration~$i\geq 1$:
	\begin{align}
		\label{eqn:inv-detstrat}
		&\forall L,X_u,X_c. R(L) \land \lambda(L,X_u,X_c) \Rightarrow
		\lambda_0(L,X_u,X_c).\\
		&\forall L,X_u \exists X_c. \lambda(L,X_u,X_c).\\
		&\forall j=1\ldots i-1, \forall L,X_u : \lnot \big( 
		(\exists X_c\setminus\{x_j\}. \lambda(L,X_u,X_c) \land x_j)
		\land (\exists X_c\setminus\{x_j\}. \lambda(L,X_u,X_c) \land \lnot x_j)
		\big).
	\end{align}
	In words, the first invariant says that at all states in~$R(L)$, the partial
	strategy computed so far is compatible with~$\lambda_0$; and the second
	invariant says that~$\lambda$ is satisfiable given any~$L,X_u$.
	This will ensure that $\lambda$ is always a \emph{winning} quasi-strategy.
	The third invariant
	states the functionality of~$\lambda$ for the first $i-1$ variables. In fact,
	it states that, given $L,X_u$, there is only one possible value of~$x_j$ that
	satisfies~$\lambda$. Thus, at the end of the algorithm, these invariants will
	yield that~$\lambda$ is a function compatible with~$\lambda_0$ which is what
	we want.

	The claim holds trivially for $i = 1$. Consider~$i\geq 2$.
	On lines $3$ and~$4$, we define $f_x(L,X_u)$ (resp. $f_{\bar{x}}(L,X_u)$), the subset of $L,X_u$ on which
	$x_i$ can be set to true (resp. false) by~$\lambda$.
	On line~$5$, the set \textit{care} is defined as the set $L,X_u$ where $R(L)$ holds, and
	$x_j$ can only be set to either true or false. Intuitively, $\lambda$ must be defined
	uniquely on this set, whereas it can be defined arbitrarily outside. In fact,
	outside~$R(L)$ we do not care about~$\lambda$ since it does not matter for
	winning; and outside~$(\lnot f_x(L,X_u) \lor \lnot f_{\overline{x}}(L,X_u))$,
	we know that~$x_j$ can take both values.
	On line~$6$, we define the function~$g_x(L,X_u)$ by the restrict operator
	$\Downarrow$, which gives an arbitrary function compatible with $f_x(L,X_u)$
	on the set $\textit{care}(L,X_u)$. This means that~$x_i$ is set to~$1$ when
	$R(L) \land f_x(L,X_u)$ holds, and to~$0$ when $R(L) \land f_{\bar{x}}(L,X_u)$
	holds. It follows that, by construction, the updated~$\lambda$ is still compatible
	with~$\lambda_0$. Moreover, since~$g_x(L,X_u)$ is a function, $\lambda$
	is also functional on variables $x_1,\ldots,x_i$.
\end{proof}

Finally, we present a possible further optimization. One could execute the
algorithm once, recover the new strategy $\lambda'$ and re-run the algorithm
with $R \equiv \calR(\lambda')$, which is clearly a winning region of \eve.
This would still guarantee the invariants hold and is thus sound.

\section{Experimental results}
\label{sec:experiments-discussion}
\begin{figure}[htb]
  \begin{minipage}{0.5\textwidth}
    \includegraphics[width=\textwidth]{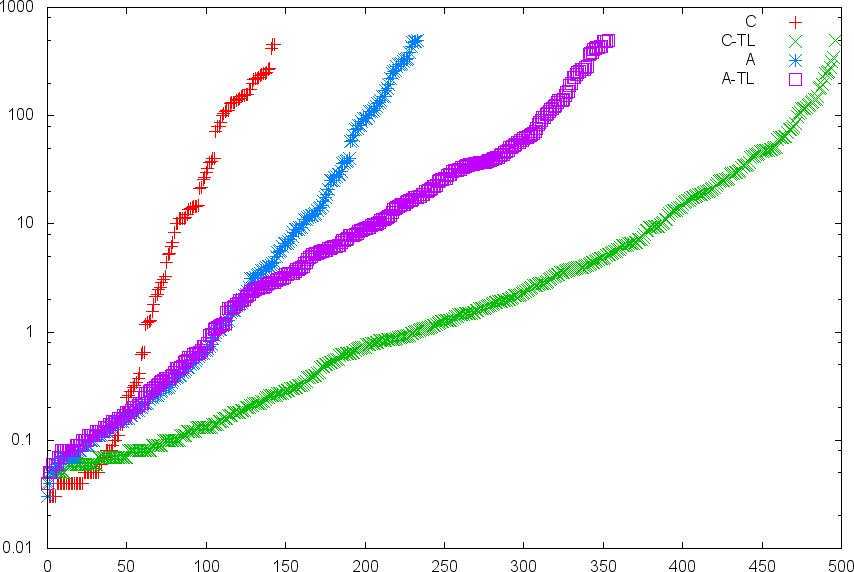}
    \caption{Time (in seconds) to check realizability.}
    \label{fig:real-sorted}
  \end{minipage}
  \begin{minipage}{0.5\textwidth}
  \includegraphics[width=\textwidth]{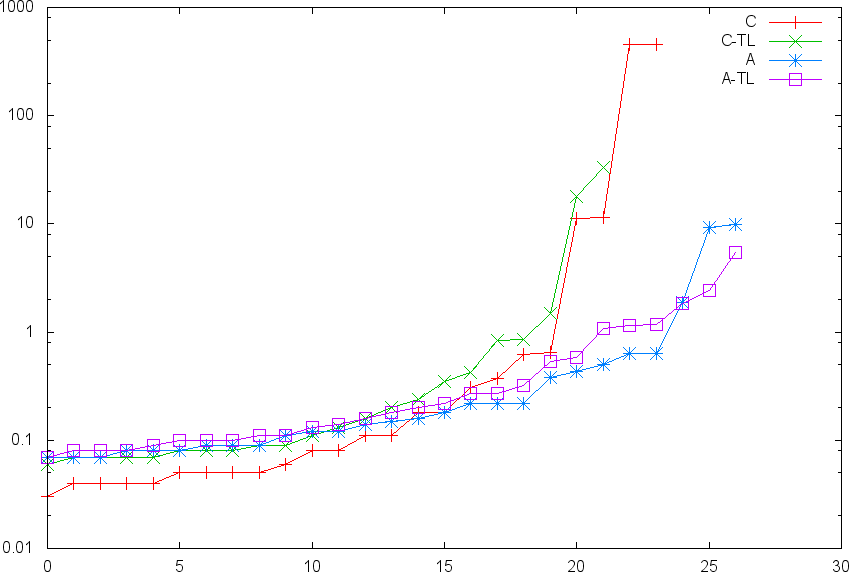}
  \caption{Time (in seconds) for \emph{cnt} benchmarks.}
  \label{fig:real-cnt}
  \end{minipage}
\end{figure}
\begin{wrapfigure}{r}{0.5\textwidth}
  \vspace{-4em}
  \includegraphics[width=0.5\textwidth]{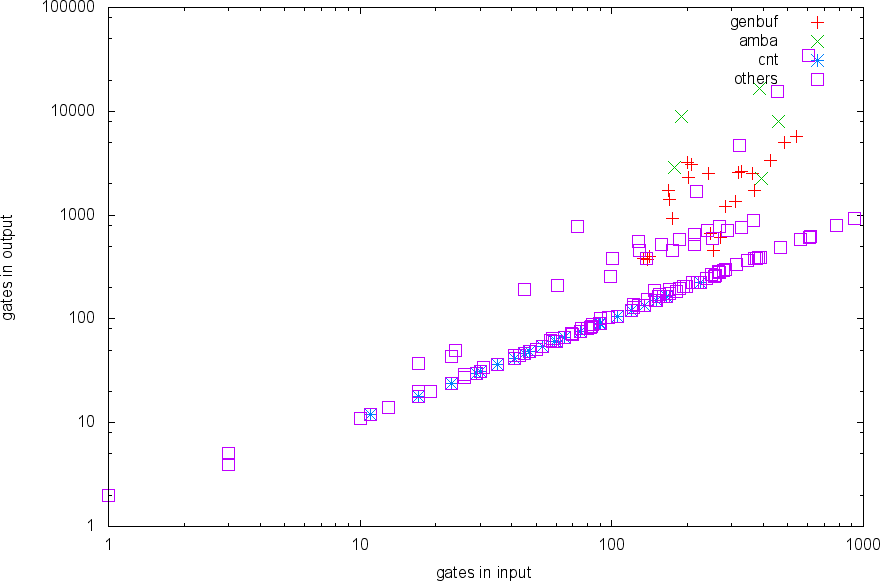}
  \caption{Size of the synthesized strategy.}
  \label{fig:synth}
\end{wrapfigure}
We evaluated four different algorithms: 
\begin{inparaenum}
\item[(C)] the classical fixpoint computation with a precomputed
  transition relation;
\item[(C-TL)] the classical fixpoint computation using the partitioned transition relation;
\item[(A)] the algorithm~\ref{alg:abs_synth} with a precomputed abstract
  transition relation;
\item[(A-TL)] the algorithm~\ref{alg:abs_synth} using abstract operators
  implemented to avoid using a transition relation (this implies
  $\overline{\post}$ was used instead of the exact $\post$
  operator).
\end{inparaenum}
The benchmarks that we used for the evaluation are provided for the SyntComp
(Synthesis Competition) \texttt{https://syntcompdb.iaik.tugraz.at/}.  At the time
of submission of this work, there were $432$ benchmarks provided by the
organizers of the competition. We have submitted $102$ additional benchmarks. 

\figurename~\ref{fig:real-sorted} summarizes performances of the algorithms on
all our benchmarks.  The horizontal axis is the number of instances that can be
solved within the time limit given by the vertical axis.  In general C-TL
performs better, however the algorithms that use abstraction perform better on
some examples.  This is in particular the case for the ``cnt'' benchmark, as can
be seen in \figurename~\ref{fig:real-cnt}.  In these benchmarks, there is a
counter (its size depends on the benchmark), the adversary can increment it and
the controller should reset it at the right moment.  The set of reachable states
is enormous but the winning strategy is quite simple which may explain why
abstraction works better. The abstract algorithms were able to solve more of
these examples within the time limit of $500$s that we fixed.

In \figurename~\ref{fig:synth} we compared the size of the synthesized circuits
with the size of the input circuits for the different sets of benchmarks.  Most
of the time, our method allows to find solutions that are not too big when
compared to the input circuit.

It is worth mentioning that from the $534$ benchmarks considered, we were able
to determine realizability in under $500$s for all but less than $30$ of them.
Amongst these, $369$ are known to be realizable. We were able to synthesize a
circuit, again in under $500$s, for all but less than $35$.

\section{Acknowledgements}
We thank Robert K\"{o}nighofer for providing us their implementation of the
classic fixpoint computation algorithm as well as a benchmarking framework for
it. This implementation~\cite{bks14} was the starting point for our tool.

\bibliographystyle{eptcs}
\bibliography{syntcomp}{}

\begin{thebibliography}{10}
\providecommand{\bibitemdeclare}[2]{}
\providecommand{\surnamestart}{}
\providecommand{\surnameend}{}
\providecommand{\urlprefix}{Available at }
\providecommand{\url}[1]{\texttt{#1}}
\providecommand{\href}[2]{\texttt{#2}}
\providecommand{\urlalt}[2]{\href{#1}{#2}}
\providecommand{\doi}[1]{doi:\urlalt{http://dx.doi.org/#1}{#1}}
\providecommand{\bibinfo}[2]{#2}

\bibitemdeclare{article}{ap07}
\bibitem{ap07}
\bibinfo{author}{Luca \surnamestart de~Alfaro\surnameend} \&
  \bibinfo{author}{Pritam \surnamestart Roy\surnameend} (\bibinfo{year}{2010}):
  \emph{\bibinfo{title}{Solving games via three-valued abstraction
  refinement}}.
\newblock {\sl \bibinfo{journal}{Information and Computation}}
  \bibinfo{volume}{208}(\bibinfo{number}{6}), pp. \bibinfo{pages}{666--676},
  \doi{10.1016/j.ic.2009.05.007}.

\bibitemdeclare{inproceedings}{bks14}
\bibitem{bks14}
\bibinfo{author}{Roderick \surnamestart Bloem\surnameend},
  \bibinfo{author}{Robert \surnamestart K{\"o}nighofer\surnameend} \&
  \bibinfo{author}{Martina \surnamestart Seidl\surnameend}
  (\bibinfo{year}{2014}): \emph{\bibinfo{title}{SAT-Based Synthesis Methods for
  Safety Specs}}.
\newblock In: {\sl \bibinfo{booktitle}{VMCAI}}, {\sl \bibinfo{series}{LNCS}}
  \bibinfo{volume}{8318}, \bibinfo{publisher}{Springer}, pp.
  \bibinfo{pages}{1--20}, \doi{10.1007/978-3-642-54013-4\_1}.

\bibitemdeclare{article}{bryant86}
\bibitem{bryant86}
\bibinfo{author}{Randal~E. \surnamestart Bryant\surnameend}
  (\bibinfo{year}{1986}): \emph{\bibinfo{title}{Graph-based algorithms for
  boolean function manipulation}}.
\newblock {\sl \bibinfo{journal}{Computers, IEEE Transactions on}}
  \bibinfo{volume}{100}(\bibinfo{number}{8}), pp. \bibinfo{pages}{677--691},
  \doi{10.1109/TC.1986.1676819}.

\bibitemdeclare{inproceedings}{bcl91}
\bibitem{bcl91}
\bibinfo{author}{Jerry~R. \surnamestart Burch\surnameend},
  \bibinfo{author}{Edmund~M. \surnamestart Clarke\surnameend} \&
  \bibinfo{author}{David~E. \surnamestart Long\surnameend}
  (\bibinfo{year}{1991}): \emph{\bibinfo{title}{Symbolic Model Checking with
  Partitioned Transistion Relations}}.
\newblock In: {\sl \bibinfo{booktitle}{VLSI}}, pp. \bibinfo{pages}{49--58}.

\bibitemdeclare{inproceedings}{cgjlv00}
\bibitem{cgjlv00}
\bibinfo{author}{Edmund \surnamestart Clarke\surnameend}, \bibinfo{author}{Orna
  \surnamestart Grumberg\surnameend}, \bibinfo{author}{Somesh \surnamestart
  Jha\surnameend}, \bibinfo{author}{Yuan \surnamestart Lu\surnameend} \&
  \bibinfo{author}{Helmut \surnamestart Veith\surnameend}
  (\bibinfo{year}{2000}): \emph{\bibinfo{title}{Counterexample-guided
  abstraction refinement}}.
\newblock In: {\sl \bibinfo{booktitle}{CAV}}, {\sl \bibinfo{series}{LNCS}}
  \bibinfo{volume}{1855}, \bibinfo{publisher}{Springer}, pp.
  \bibinfo{pages}{154--169}, \doi{10.1007/10722167\_15}.

\bibitemdeclare{inproceedings}{cgt04}
\bibitem{cgt04}
\bibinfo{author}{Edmund \surnamestart Clarke\surnameend}, \bibinfo{author}{Orna
  \surnamestart Grumberg\surnameend}, \bibinfo{author}{Muralidhar \surnamestart
  Talupur\surnameend} \& \bibinfo{author}{Dong \surnamestart Wang\surnameend}
  (\bibinfo{year}{2003}): \emph{\bibinfo{title}{High level verification of
  control intensive systems using predicate abstraction}}.
\newblock In: {\sl \bibinfo{booktitle}{MEMOCODE}}, \bibinfo{publisher}{IEEE},
  pp. \bibinfo{pages}{55--64}, \doi{10.1109/MEMCOD.2003.1210089}.

\bibitemdeclare{inproceedings}{cbm90}
\bibitem{cbm90}
\bibinfo{author}{Olivier \surnamestart Coudert\surnameend},
  \bibinfo{author}{Christian \surnamestart Berthet\surnameend} \&
  \bibinfo{author}{Jean~Christophe \surnamestart Madre\surnameend}
  (\bibinfo{year}{1990}): \emph{\bibinfo{title}{Verification of synchronous
  sequential machines based on symbolic execution}}.
\newblock In: {\sl \bibinfo{booktitle}{Automatic verification methods for
  finite state systems}}, {\sl \bibinfo{series}{LNCS}} \bibinfo{volume}{407},
  \bibinfo{publisher}{Springer}, pp. \bibinfo{pages}{365--373},
  \doi{10.1007/3-540-52148-8\_30}.

\bibitemdeclare{inproceedings}{cm90}
\bibitem{cm90}
\bibinfo{author}{Olivier \surnamestart Coudert\surnameend} \&
  \bibinfo{author}{Jean~Christophe \surnamestart Madre\surnameend}
  (\bibinfo{year}{1990}): \emph{\bibinfo{title}{A Unified Framework for the
  Formal Verification of Sequential Circuits}}.
\newblock In: {\sl \bibinfo{booktitle}{ICCAD}}, pp. \bibinfo{pages}{126--129}.

\bibitemdeclare{inproceedings}{cmb91}
\bibitem{cmb91}
\bibinfo{author}{Olivier \surnamestart Coudert\surnameend},
  \bibinfo{author}{Jean~Christophe \surnamestart Madre\surnameend} \&
  \bibinfo{author}{Christian \surnamestart Berthet\surnameend}
  (\bibinfo{year}{1991}): \emph{\bibinfo{title}{Verifying temporal properties
  of sequential machines without building their state diagrams}}.
\newblock In: {\sl \bibinfo{booktitle}{CAV}}, {\sl \bibinfo{series}{LNCS}}
  \bibinfo{volume}{531}, \bibinfo{publisher}{Springer}, pp.
  \bibinfo{pages}{23--32}, \doi{10.1007/BFb0023716}.

\bibitemdeclare{inproceedings}{cc77}
\bibitem{cc77}
\bibinfo{author}{Patrick \surnamestart Cousot\surnameend} \&
  \bibinfo{author}{Radhia \surnamestart Cousot\surnameend}
  (\bibinfo{year}{1977}): \emph{\bibinfo{title}{Abstract interpretation: a
  unified lattice model for static analysis of programs by construction or
  approximation of fixpoints}}.
\newblock In: {\sl \bibinfo{booktitle}{POPL}}, \bibinfo{publisher}{ACM}, pp.
  \bibinfo{pages}{238--252}, \doi{10.1145/512950.512973}.

\bibitemdeclare{inproceedings}{ehlers10}
\bibitem{ehlers10}
\bibinfo{author}{R{\"u}diger \surnamestart Ehlers\surnameend}
  (\bibinfo{year}{2010}): \emph{\bibinfo{title}{Symbolic Bounded Synthesis}}.
\newblock In: {\sl \bibinfo{booktitle}{CAV}}, {\sl \bibinfo{series}{LNCS}}
  \bibinfo{volume}{6174}, \bibinfo{publisher}{Springer}, pp.
  \bibinfo{pages}{365--379}, \doi{10.1007/s10703-011-0137-x}.

\bibitemdeclare{inproceedings}{ej91}
\bibitem{ej91}
\bibinfo{author}{E.~Allen \surnamestart Emerson\surnameend} \&
  \bibinfo{author}{Charanjit~S. \surnamestart Jutla\surnameend}
  (\bibinfo{year}{1991}): \emph{\bibinfo{title}{Tree automata, mu-calculus and
  determinacy}}.
\newblock In: {\sl \bibinfo{booktitle}{FOCS}}, \bibinfo{organization}{IEEE},
  pp. \bibinfo{pages}{368--377}, \doi{10.1109/SFCS.1991.185392}.

\bibitemdeclare{inproceedings}{fjr09}
\bibitem{fjr09}
\bibinfo{author}{Emmanuel \surnamestart Filiot\surnameend},
  \bibinfo{author}{Naiyong \surnamestart Jin\surnameend} \&
  \bibinfo{author}{Jean-Fran{\c{c}}ois \surnamestart Raskin\surnameend}
  (\bibinfo{year}{2009}): \emph{\bibinfo{title}{An Antichain Algorithm for
  {LTL} Realizability}}.
\newblock In: {\sl \bibinfo{booktitle}{{CAV}}}, {\sl \bibinfo{series}{{LNCS}}}
  \bibinfo{volume}{5643}, \bibinfo{publisher}{Springer}, pp.
  \bibinfo{pages}{263--277}, \doi{10.1007/978-3-642-02658-4\_22}.

\bibitemdeclare{inproceedings}{gradel04}
\bibitem{gradel04}
\bibinfo{author}{Erich \surnamestart Gr{\"a}del\surnameend}
  (\bibinfo{year}{2004}): \emph{\bibinfo{title}{Positional Determinacy of
  Infinite Games}}.
\newblock In: {\sl \bibinfo{booktitle}{STACS}}, {\sl \bibinfo{series}{LNCS}}
  \bibinfo{volume}{2996}, \bibinfo{publisher}{Springer}, pp.
  \bibinfo{pages}{4--18}, \doi{10.1007/978-3-540-24749-4\_2}.

\bibitemdeclare{inproceedings}{gs97}
\bibitem{gs97}
\bibinfo{author}{Susanne \surnamestart Graf\surnameend} \&
  \bibinfo{author}{Hassen \surnamestart Sa{\"\i}di\surnameend}
  (\bibinfo{year}{1997}): \emph{\bibinfo{title}{Construction of abstract state
  graphs with PVS}}.
\newblock In: {\sl \bibinfo{booktitle}{CAV}}, {\sl \bibinfo{series}{LNCS}}
  \bibinfo{volume}{1254}, \bibinfo{publisher}{Springer}, pp.
  \bibinfo{pages}{72--83}, \doi{10.1007/3-540-63166-6\_10}.

\bibitemdeclare{inproceedings}{hjm03}
\bibitem{hjm03}
\bibinfo{author}{Thomas~A. \surnamestart Henzinger\surnameend},
  \bibinfo{author}{Ranjit \surnamestart Jhala\surnameend} \&
  \bibinfo{author}{Rupak \surnamestart Majumdar\surnameend}
  (\bibinfo{year}{2003}): \emph{\bibinfo{title}{Counterexample-guided
  control}}.
\newblock In: {\sl \bibinfo{booktitle}{ICALP}}, {\sl \bibinfo{series}{LNCS}}
  \bibinfo{volume}{2719}, \bibinfo{publisher}{Springer}, pp.
  \bibinfo{pages}{886--902}, \doi{10.1007/3-540-45061-0\_69}.

\bibitemdeclare{inproceedings}{hmmr00}
\bibitem{hmmr00}
\bibinfo{author}{Thomas~A. \surnamestart Henzinger\surnameend},
  \bibinfo{author}{Rupak \surnamestart Majumdar\surnameend},
  \bibinfo{author}{Freddy Y.~C. \surnamestart Mang\surnameend} \&
  \bibinfo{author}{Jean-Fran\c{c}ois \surnamestart Raskin\surnameend}
  (\bibinfo{year}{2000}): \emph{\bibinfo{title}{Abstract Interpretation of Game
  Properties}}.
\newblock In: {\sl \bibinfo{booktitle}{SAS}}, pp. \bibinfo{pages}{220--239},
  \doi{10.1007/978-3-540-45099-3\_12}.

\bibitemdeclare{inproceedings}{hbbm97}
\bibitem{hbbm97}
\bibinfo{author}{Youpyo \surnamestart Hong\surnameend},
  \bibinfo{author}{Peter~A \surnamestart Beerel\surnameend},
  \bibinfo{author}{Jerry~R \surnamestart Burch\surnameend} \&
  \bibinfo{author}{Kenneth~L \surnamestart McMillan\surnameend}
  (\bibinfo{year}{1997}): \emph{\bibinfo{title}{Safe BDD minimization using
  don't cares}}.
\newblock In: {\sl \bibinfo{booktitle}{Proceedings of the 34th annual Design
  Automation Conference}}, \bibinfo{organization}{ACM}, pp.
  \bibinfo{pages}{208--213}, \doi{10.1145/266021.266068}.

\bibitemdeclare{inproceedings}{kurshan94}
\bibitem{kurshan94}
\bibinfo{author}{Robert~P. \surnamestart Kurshan\surnameend}
  (\bibinfo{year}{1994}): \emph{\bibinfo{title}{Automata-theoretic verification
  of coordinating processes}}.
\newblock In: {\sl \bibinfo{booktitle}{11th International Conference on
  Analysis and Optimization of Systems Discrete Event Systems}},
  \bibinfo{organization}{Springer}, pp. \bibinfo{pages}{16--28},
  \doi{10.1007/BFb0033528}.

\bibitemdeclare{article}{mccluskey56}
\bibitem{mccluskey56}
\bibinfo{author}{Edward~J. \surnamestart McCluskey\surnameend}
  (\bibinfo{year}{1956}): \emph{\bibinfo{title}{Minimization of Boolean
  Functions*}}.
\newblock {\sl \bibinfo{journal}{Bell system technical Journal}}
  \bibinfo{volume}{35}(\bibinfo{number}{6}), pp. \bibinfo{pages}{1417--1444},
  \doi{10.1002/j.1538-7305.1956.tb03835.x}.

\bibitemdeclare{inproceedings}{nlbrw14}
\bibitem{nlbrw14}
\bibinfo{author}{Nina \surnamestart Narodytska\surnameend},
  \bibinfo{author}{Alexander \surnamestart Legg\surnameend},
  \bibinfo{author}{Fahiem \surnamestart Bacchus\surnameend},
  \bibinfo{author}{Leonid \surnamestart Ryzhyk\surnameend} \&
  \bibinfo{author}{Adam \surnamestart Walker\surnameend}
  (\bibinfo{year}{2014}): \emph{\bibinfo{title}{Solving Games without
  Controllable Predecessor}}.
\newblock In: {\sl \bibinfo{booktitle}{CAV}}, \bibinfo{organization}{Springer}.

\bibitemdeclare{article}{py86}
\bibitem{py86}
\bibinfo{author}{Christos~H. \surnamestart Papadimitriou\surnameend} \&
  \bibinfo{author}{Mihalis \surnamestart Yannakakis\surnameend}
  (\bibinfo{year}{1986}): \emph{\bibinfo{title}{A note on succinct
  representations of graphs}}.
\newblock {\sl \bibinfo{journal}{Information and Control}}
  \bibinfo{volume}{71}(\bibinfo{number}{3}), pp. \bibinfo{pages}{181 -- 185},
  \doi{10.1016/S0019-9958(86)80009-2}.

\bibitemdeclare{article}{sw96}
\bibitem{sw96}
\bibinfo{author}{Martin \surnamestart Sauerhoff\surnameend} \&
  \bibinfo{author}{Ingo \surnamestart Wegener\surnameend}
  (\bibinfo{year}{1996}): \emph{\bibinfo{title}{On the complexity of minimizing
  the OBDD size for incompletely specified functions}}.
\newblock {\sl \bibinfo{journal}{IEEE Trans. on CAD of Integrated Circuits and
  Systems}} \bibinfo{volume}{15}(\bibinfo{number}{11}), pp.
  \bibinfo{pages}{1435--1437}, \doi{10.1109/43.543775}.

\bibitemdeclare{article}{ss09}
\bibitem{ss09}
\bibinfo{author}{Saqib \surnamestart Sohail\surnameend} \&
  \bibinfo{author}{Fabio \surnamestart Somenzi\surnameend}
  (\bibinfo{year}{2009}): \emph{\bibinfo{title}{Safety first: A two-stage
  algorithm for LTL games}}.
\newblock {\sl \bibinfo{journal}{FMCAD}}, pp. \bibinfo{pages}{77--84},
  \doi{10.1007/s10009-012-0224-3}.

\bibitemdeclare{incollection}{somenzi99}
\bibitem{somenzi99}
\bibinfo{author}{Fabio \surnamestart Somenzi\surnameend}
  (\bibinfo{year}{1999}): \emph{\bibinfo{title}{Binary Decision Diagrams}}.
\newblock In: {\sl \bibinfo{booktitle}{Calculational system design}},
  \bibinfo{volume}{173}, \bibinfo{publisher}{IOS Press}, p.
  \bibinfo{pages}{303}.

\bibitemdeclare{article}{tarski55}
\bibitem{tarski55}
\bibinfo{author}{Alfred \surnamestart Tarski\surnameend} et~al.
  (\bibinfo{year}{1955}): \emph{\bibinfo{title}{A lattice-theoretical fixpoint
  theorem and its applications}}.
\newblock {\sl \bibinfo{journal}{Pacific journal of Mathematics}}
  \bibinfo{volume}{5}(\bibinfo{number}{2}), pp. \bibinfo{pages}{285--309},
  \doi{10.2140/pjm.1955.5.285}.

\bibitemdeclare{inproceedings}{thomas95}
\bibitem{thomas95}
\bibinfo{author}{Wolfgang \surnamestart Thomas\surnameend}
  (\bibinfo{year}{1995}): \emph{\bibinfo{title}{On the synthesis of strategies
  in infinite games}}.
\newblock In: {\sl \bibinfo{booktitle}{STACS}},
  \bibinfo{organization}{Springer}, pp. \bibinfo{pages}{1--13},
  \doi{10.1007/3-540-59042-0\_57}.

\end{thebibliography}

\end{document}